\documentclass[11pt,journal,onecolumn]{IEEEtran}

\title{One-Bit Sigma-Delta modulation on the circle}

\author{Olga Graf, Felix Krahmer, and Sara Krause-Solberg}

\usepackage[utf8]{inputenc}
\usepackage{graphicx}
\usepackage{amsmath,amsfonts,amssymb,geometry,xcolor,enumitem,graphicx,subfigure,epsfig,amsthm,mathtools}
\usepackage{amscd}
\usepackage{enumerate}
\usepackage{lipsum}
\usepackage[export]{adjustbox}

\allowdisplaybreaks

\newcommand\blfootnote[1]{%
	\begingroup
	\renewcommand\thefootnote{}\footnote{#1}%
	\addtocounter{footnote}{-1}%
	\endgroup
}

\newtheorem{theorem}{Theorem}
\newtheorem{lemma}{Lemma}

\newtheorem{proposition}{\bf Proposition}

\def\supp{\operatorname{supp}}

\begin{document}
	
\maketitle

\begin{abstract}
Manifold models in data analysis and signal processing have become more prominent in recent years. In this paper, we will look at one of the main tasks of modern signal processing, namely, at analog-to-digital (A/D) conversion in connection with a simple manifold model - the circle. We will focus on Sigma-Delta modulation which is a popular method for A/D conversion of bandlimited signals that employs coarse quantization coupled with oversampling. Classical  Sigma-Delta schemes provide mismatches and large errors at the initialization point if the signal to be converted is defined on the circle. In this paper, our goal is to get around these problems for Sigma-Delta schemes. Our results show how to design an update for the first and the second order schemes based on the reconstruction error analysis such that for the updated scheme the reconstruction error is improved.
\end{abstract}


\blfootnote{This work was presented at IEEE Statistical Signal Processing Workshop (SSP), Freiburg, Germany, 2018 \cite{KSGK2018}, and at the 13th Intl. Conf. on Sampling Theory and Applications (SampTA), Bordeaux, France, 2019 \cite{GKKS2019}.
	
The authors are with the Dept. of Mathematics, Technical University of Munich, Boltzmannstraße 3, 85748 Garching/Munich, Germany, Email: \{olga.graf, felix.krahmer\}@ma.tum.de and with the Institute of Mathematics, Hamburg University of Technology, Am Schwarzenberg-Campus 1, 21073 Hamburg, Germany, Email: s.krause-solberg@tuhh.de.
	
The support by the German Science Foundation (DFG) under Grant SFB Transregio 109, Discretization in Geometry and Dynamics, is gratefully acknowledged by the authors.}

\section{Introduction}

\subsection{Analog-to-digital conversion}

Whenever data is processed using computers, an analog-to-digital (A/D) conversion is used. It comprises two stages: sampling and quantization. First of all, sampling converts an analog signal into a discrete signal. In the process of quantization, each value of this discrete signal is mapped to some element from a finite alphabet, thus becoming amenable to digital storage. The extreme case when the quantization alphabet consists of only two elements is called \textit{1-bit quantization}. Quantization maps can vary from simple memoryless quantization, where each value is replaced with an element of the alphabet that best approximates it, to more complicated maps based on recurrence relations.

Digital-to-analog (D/A) conversion, a reverse process during which  the original signal is reconstructed from quantized values, is usually carried out by applying an appropriate low-pass filter. The accuracy of such reconstruction serves as a main criteria for the quality assessment of the quantization scheme.

\subsection{Sigma-Delta modulation in mathematical literature}

In this paper, we will focus on \textit{Sigma-Delta} ($\varSigma\varDelta$) \textit{modulation} which is a popular method for the quantization of bandlimited functions. The underlying idea of this method is to compute the quantized values recursively via a difference equation such that the low frequency content of the quantized representation approximates the signal. On the other hand, the quantization error is designed to approximately form a high-pass sequence (so called “noise shaping”) and is later removed to a large extent in the reconstruction procedure via a low-pass filter. The power of this quantization method lies in the fact that it makes the design of cheap analog circuits of low complexity possible; namely, it allows for high accuracy with an arbitrarily coarse quantization alphabet at the expense of oversampling.

$\Sigma\Delta$ modulation has been known to circuit engineers since the 1963 pioneering work \cite{Inose1963} of Inose and Yasuda. It was later followed by a rigorous mathematical study initiated by Daubechies and DeVore \cite{Daubechies2003} in the early 2000's. Since then, the mathematical literature on this method of quantization has grown rapidly. Early papers on $\Sigma\Delta$ modulation dealt with the classical case of bounded bandlimited functions on the real line and focused on the reconstruction accuracy as a function of oversampling rate \cite{Gunturk2003,Deift2011,KrahmerWard2012,CD02}. Later, $\Sigma\Delta$ modulation schemes were extended to finite frame expansions in \cite{Krahmer2012,Benedetto2006,BenedettoPowell2006}. A number of works also study  $\Sigma\Delta$ modulation in combination with compressed sensing \cite{Krahmer2014,Gunturk2013,Saab2018}. For an overview of $\Sigma\Delta$ modulation in various settings and more general classes of noise shaping methods, we refer the reader to \cite{CGKSY15}.

Despite the growing importance of manifold models in data and signal processing, there is still little quantization literature available for such models, both for $\Sigma\Delta$ modulation and other quantization methods.
For the special case of Grassmannian manifolds, which arise naturally in wireless communications, there appeared some papers studying their quantization properties \cite{Modal2007}. Some results for manifolds were obtained in connection with 1-bit compressed sensing. Namely, the problem of recovering an unknown data point on a given manifold from 1-bit quantized random measurements was studied in \cite{IKKSM2018}. Recently, this work was further extended in \cite{ILNS2019} to incorporate $\Sigma\Delta$ modulation schemes.

\subsection{Our motivation and contribution}

In this work, we aim to provide a first step towards a theory for $\Sigma\Delta$ modulation to represent functions defined on a closed manifolds via its quantized samples.
Our work is application-driven: there exists a connection between 1-bit $\Sigma\Delta$ modulation and a printing technique known as digital halftoning \cite{Kite1997}. In fact, halftoning can be done by means of 1-bit $\Sigma\Delta$ modulation. The emerging 3D printing technologies allow for applying halftoning algorithms on various closed surfaces \cite{MSUA2017}. In order to pave the way to employ halftoning via $\Sigma\Delta$ modulation for printing on closed surfaces, we will start with the problem of $\Sigma\Delta$ modulation on the simplest closed one-dimensional manifold.

More precisely, we generalize the classical $\Sigma\Delta$ modulation of bandlimited functions on the real line to the case of bandlimited functions on the one-dimensional torus.
We show that $\Sigma\Delta$ schemes encounter mismatches when applied to functions defined on such domain. In order to avoid these mismatches and to improve the reconstruction error, we propose an update, namely, a small uniform shift by an appropriate amount applied to the input function. Our motivation is that a small constant shift of the whole image will hardly be visible, whereas spatially varying errors of the same size can be perceived as artefacts.

Our contributions are twofold:

1) First, we provide the error analysis for the $m$-th order $\Sigma\Delta$ schemes on the circle (see Proposition 2).

2) Motivated by the error analysis, we show how to find an update (uniform shift) for the first and the second order $\Sigma\Delta$ schemes such that the shifted function is recovered with an accuracy of $O(N^{-1})$ and $O(N^{-2})$, respectively (see Theorems 1, 2), where $N$ is the number of samples. The second order scheme is of special interest as in general case it is not possible to obtain $O(N^{-2})$ for this scheme without update.

\subsection{Organization of the paper}

The paper is organized as follows. In Section II, we review the basics of $\Sigma\Delta$ modulation as well as sampling, quantization and reconstruction of bandlimited functions whose domain is the unit circle. In Section III, we give the error analysis for the $m$-th order $\Sigma\Delta$ schemes on the circle. In Section IV, we state our main result: a uniform update for the first and the second order $\Sigma\Delta$ scheme. We also briefly discuss how one can possibly find updates for higher order schemes. In Section V, we support our theoretical findings with results from numerical simulations and conclude with suggestions for future work in Section VI.

\section{Preliminaries}

\subsection{Basics of $\varSigma\varDelta$ modulation}

Given a (finite or infinite) sequence of samples $(y_{n})_{n\in\mathbb{Z}}$, the \emph{1-bit $\varSigma\varDelta$ quantizer} runs the following
iteration for $n\in\mathbb{Z}$:
\begin{equation}\label{sigmadelta}
	\begin{split}
		v_n&=(h\ast v)_{n}+y_n-q_{n},\\
		q_n&=\mathrm{sign}((h\ast v)_{n}+y_n),
	\end{split}
\end{equation}
where $q_{n}$ are the quantized values, $v_{n}$ are the state variables, $h$ is the feedback filter with $k$ tabs described by the recurrence relation $(h\ast v)_{n}:=\sum_{j=0}^{k}h_{j}v_{n-j}$, and
\begin{align*}
	\mathrm{sign}(x):=\begin{cases}
		1, &x>0,\\
		-1, &x\leq 0.\\
	\end{cases}
\end{align*}
As this choice of $q_{n}$ minimizes $|v_{n}|$ at each time instance, the second line in \eqref{sigmadelta} is referred to as the \emph{greedy quantization rule}.

Sometimes, it is useful to rewrite the first line in \eqref{sigmadelta} in terms of another state variable $u_n$. If $y_n-q_n$ can be rewritten as the $m$-th order backward finite difference of some bounded sequence $u_n$, i.e., 
\vspace{-0.2cm}
\begin{equation}\label{sigmadelta:m}
	\Delta^m u_n=y_n-q_n,
\end{equation}
then it is said that \eqref{sigmadelta} is the $m$\emph{-th order} 1-bit $\Sigma\Delta$ quantizer. Namely, if we require that for some finitely supported $g$,
\begin{equation}\label{sigmadelta:m2}
\Delta^m g=\delta^0-h,
\end{equation}
where $\delta^0$ denotes the Kronecker delta, the new state variable $u_n$ which satisfies \eqref{sigmadelta:m} can be written in terms of $v_n$ as
\begin{equation}\label{sigmadelta:m3}
u_n=(g\ast v)_{n}.
\end{equation}

In this paper, we assume that all feedback filters $h$ are taken from the class of filters with minimal support, i.e., $|\supp h|=m$ (see \cite{Gunturk2003,Deift2011} for details). In particular, for the second order filters with $k$ tabs, we consider structure
\begin{equation}\label{filter2}
h=(0,h_1,0,\dots,0,h_k),
\end{equation}
where $h_1=k/(k-1)$ and $h_k=1-h_1$.

For a  $\Sigma\Delta$ quantizer to be useful in practice it should be stable. Here, stability means that there exists a number $\mu>0$ such that for each bounded input sequence $(y_{n})_{n}$ with $\|y\|_{\ell^{\infty}}\leq \mu$, the sequence of state variables $(v_{n})_{n}$ (or $(u_{n})_{n}$) is also a bounded sequence. The first order $\Sigma\Delta$ quantizer is known to be stable. In case of higher order schemes, the feedback filter should be chosen carefully so that stability is guaranteed.


The following stability criterion provides a sufficient condition for the stability of a $\Sigma\Delta$ scheme \eqref{sigmadelta}. This criterion is well-known to the engineering community (e.g., \cite{SchSn1991}), but had not been used in a rigorous mathematical framework until the work of G\"unt\"urk \cite{Gunturk2003}.

\begin{proposition}[Stability Criterion]\cite{Gunturk2003,SchSn1991}
	Consider a $\varSigma\varDelta$ quantizer given by the recurrence relation \eqref{sigmadelta}. If
	\begin{equation}\label{eq:stability} 
		\|h\|_{\ell^1} \leq 2-\|y\|_{\ell^{\infty}},
	\end{equation}
	then the modulator is stable for all inputs $y_n$.
\end{proposition}
For the first order quantizer (i.e., with the feedback filter $h=(0,1)$), Proposition 1 implies $\|y\|_{\ell^{\infty}}\leq 1$. Although not crucial for stability, this bound guarantees reasonable reconstruction from quantized values with range $[-1,1]$ as in \eqref{sigmadelta}.

The reconstruction error analysis of $\Sigma\Delta$ modulation typically assumes stability as it relies on the boundedness of the state variables $v_{n}$ (or $u_{n}$). A reconstruction error bound for 1-bit $\Sigma\Delta$ quantizer in \cite{Daubechies2003} establishes a polynomial error decay of $O(\lambda^{-m})$, where $\lambda$ is a factor, by which the function is oversampled with respect to its Nyquist frequency. Later, an exponential error decay of $O(2^{-r\lambda})$ with $r\approx 0.102$ was achieved in \cite{Gunturk2003,Deift2011} by combining $\Sigma\Delta$ schemes of different orders.
This is the best known error decay rate, which is also known to be optimal \cite{KrahmerWard2012,CD02}. In this paper, we will follow the approach of \cite{Deift2011}, but only for fixed orders, hence expecting polynomial error bounds similar to those obtained in \cite{Daubechies2003}. This will be sufficient for achieving our main goal: understanding how to update the $\Sigma\Delta$ modulation scheme to make it fit our setting and assessing the improvement in accuracy of reconstruction after the update.

\subsection{Analog-to-digital conversion on the circle}

Let us briefly review sampling, quantization and reconstruction of $K$-bandlimited functions (i.e., $\supp(\widehat f)\subset[-K,K]$) $f\in L^2(\mathbb T)$ whose domain is the unit circle $\mathbb{T}=\{z\in{\mathbb{C}}:|z|=1\}$.
Uniformly sampling $f$ on a unit circle at a rate $N/2\pi$ $(N\in\mathbb{N})$ gives the samples 
\begin{equation}
	y_{n}:=
	f\left(\frac{2\pi n}{N}\right), \qquad n\in\{0,1,...,N-1\}.
\end{equation}
Note that there is a one-to-one correspondence between the samples of $2\pi$-periodic functions on the real line and the samples of functions on the unit circle. However, there will be no such correspondence between the quantized values.

Motivated by Shannon's sampling theorem, we define a reproducing kernel $\varphi^K\in L^2(\mathbb T)$ such that
\begin{align}\label{kernel:cond}
	\widehat{\varphi}^{K}(\xi)=
	\begin{cases}
		1,\quad |\xi|\leq K,\\
		0,\quad |\xi|> K.
	\end{cases}
\end{align}

By the Fourier inversion theorem and some computations, we have
\begin{equation}\label{repkernel}
	\varphi^K(x)=\frac{\sin\left((2K+1)\frac x2\right)}{\sin\left(\frac x2\right)}.
\end{equation}
In fact, \eqref{repkernel} is the Dirichlet kernel whose convolution with any $2\pi$-periodic function gives the $n$-th degree Fourier series approximation of that function. Therefore, \eqref{repkernel} is indeed the reproducing kernel for $K$-bandlimited functions on $\mathbb{T}$.

Taking the Fourier series expansion of $\widehat{f}$ and using the inverse Fourier transform yields
\begin{align}\label{shannon_circle} 
	f(t)=\frac{1}{N} \sum_{n=0}^{ N-1}f\left(\frac{2\pi n}{ N}\right) \varphi^K\left(t-\frac{2\pi n}{ N}\right)
\end{align}
for any $N=2\lambda K+1$, where $\lambda\in\mathbb R_{\ge1}$ is the oversampling parameter. Formula \eqref{shannon_circle} is the analog of Shannon's interpolation formula for functions defined on a unit circle $\mathbb{T}$.

The standard approach to function recovery from its quantized values on $\mathbb{R}$ uses the Shannon's interpolation formula, where the samples are replaced by the quantized values. For the recovery from quantized values on $\mathbb{T}$ we have, consequently,
\begin{align}\label{shannonquantized_circle}  
	f_r(t)=\frac{1}{N} \sum_{n=0}^{ N-1}q_{n} \varphi^K\left(t-\frac{2\pi n}{ N}\right),
\end{align}
where $f_r(t)$ denotes the reconstructed function.

Let the instantaneous error $e(t)$ at time $t$ be given by the pointwise difference
\begin{align}
	\left|f(t)-f_r(t)\right| = \frac{1}{N} \left|\sum_{n=0}^{ N-1}\left(y_{n}-q_{n}\right) \varphi^K\left(t-\frac{2\pi n}{ N}\right)\right|.
\end{align}
The quality of reconstruction can be measured by a variety of functional norms on $e(t)$. 
Here, we will use the norm $\|e\|_{L^{\infty}}$ which is one of the standard choices.

\section{Error analysis for $\Sigma\Delta$ schemes on a circle}

Let us start with the error analysis for the first order $\Sigma\Delta$ scheme and assume initialization ${u}_{-1}:=0$. Hereinafter we denote $\varphi_{n}^{K}:=\varphi^{K}\left(t-\frac{2\pi n}{N}\right)$.
\begin{align}\label{Eq:Rec:Err:1st}
	|f(t)-f_r(t)|&=\frac{1}{N}\left|\sum_{n=0}^{N-1}(y_n-{q}_n)\varphi_{n}^{K}\right|=\frac{1}{N}\left|\sum_{n=0}^{N-1}(u_n-{u}_{n-1})\varphi_{n}^{K}\right|\nonumber \\
	&=\frac 1N \left| \sum_{n=0}^{N-1} u_n \left(\varphi_{n}^{K} -\varphi_{n+1}^{K}\right)-u_{-1}\varphi_{0}^{K}+ u_{N-1}\varphi_{N}^{K}\right|\nonumber\\
	&=\frac 1N \left| \sum_{n=0}^{N-1} u_n \left(\varphi_{n}^{K} -\varphi_{n+1}^{K}\right)+ u_{N-1}\varphi_{N}^{K}\right|\nonumber\\
	&\leq\frac{\|u\|_{\ell^{\infty}}}{N}\sum_{n=0}^{N-1}\left|\varphi_{n}^{K} -\varphi_{n+1}^{K}\right|+\frac{|u_{N-1}|}{N}\big\|\varphi^{K}\big\|_{L^{\infty}}\nonumber\\
	&\leq\frac{\|u\|_{\ell^{\infty}}}{N}\sum_{n=0}^{N-1}\int_{t-\frac{2\pi(n+1)}{N}}^{t-\frac{2\pi n}{N}}\big|(\varphi^{K})'(y)\big|dy+\frac{|u_{N-1}|}{N}\big\|\varphi^{K}\big\|_{L^{\infty}}\nonumber\\
	&\leq \frac{\|u\|_{\ell^{\infty}}}{N}\big\|(\varphi^{K})'\big\|_{L^1}+\frac{|u_{N-1}|}{N}\big\|\varphi^{K}\big\|_{L^{\infty}}.
\end{align}

In contrast to the case of the real line with an infinite number of samples, here the boundary terms of summation by parts in \eqref{Eq:Rec:Err:1st} remain. Thus we get an additional error term depending on the last state variable $u_{N-1}$.

In order to extend this result to the $m$-th order $\Sigma\Delta$ schemes, we will need two following lemmas.

\begin{lemma}
	For the sequences $(\varphi_{n}^{K})_{n\in\mathbb N}$ and $(u_n)_{n\in\mathbb N}$ we have
	\begin{align}
		\sum_{n=0}^{N-1}\hspace{-0.1cm}\Delta^m u_n \varphi_n^{K}&= (-1)^m\sum_{n=0}^{N-1} u_n\Delta^m \varphi^{K}_{n+m}+\sum_{k=1}^m(-1)^k\Delta^{m-k}u_{-1}\Delta^{k-1}\varphi^{K}_{k-1}\nonumber\\
		&+\sum_{k=1}^m(-1)^{k+1}\Delta^{m-k}u_{N-1}\Delta^{k-1}\varphi^{K}_{N+k-1}.
	\end{align} 
\end{lemma}

\begin{proof}
	It is easy to verify that for $m=1$ we have 
	\begin{equation}\label{Eq:delta:u:phi:1}
		\sum_{n=0}^{N-1} \Delta u_n\varphi^{K}_n=-\sum_{n=0}^{N-1}u_n\bar\Delta \varphi^{K}_n-u_{-1}\varphi^{K}_0+u_{N-1}\varphi^{K}_N,
	\end{equation}
	where $\bar\Delta \varphi^{K}_n$ denotes the $m$-th order forward finite difference of $\varphi^{K}_n$. We apply \eqref{Eq:delta:u:phi:1} once to get the first equality in the following computation. Then, the second equality follows by induction or by applying \eqref{Eq:delta:u:phi:1} $m-1$ more times.	 
	\begin{align}\label{Eq:delta:u:phi:2}
	&\sum_{n=0}^{N-1}\Delta^m u_n \varphi^{K}_n = -\sum_{n=0}^{N-1}\Delta^{m-1}u_n \bar\Delta \varphi^{K}_n - \Delta^{m-1} u_{-1}\varphi^{K}_0+ \Delta^{m-1}u_{N-1}\varphi^{K}_N \nonumber\\
	&= (-1)^m\sum_{n=0}^{N-1} u_n\bar\Delta^m \varphi^{K}_n + \sum_{k=1}^m (-1)^k\Delta^{m-k}u_{-1}\bar\Delta^{k-1}\varphi^{K}_0 + \sum_{k=1}^m (-1)^{k+1} \Delta^{m-k}u_{N-1}\bar\Delta^{k-1}\varphi^{K}_N.
	\end{align}
	Noting that
	\begin{equation}\label{backforth}
	\left(\bar\Delta^m u\right)_n=\left(\Delta^m u\right)_{n+m}
	\end{equation}
	and applying it to \eqref{Eq:delta:u:phi:2} completes the proof.
\end{proof}

\begin{lemma}
	For the $k$-th order finite difference of $\varphi^{K}_k:=\varphi^{K}\left(t-\frac{2\pi k}{N}\right)$ we have
	for some $\tau\in\left(t-\frac{2\pi}{N},t\right)$
	\begin{align}\label{Eq:delta:phi}
		\Delta^{k}\varphi^{K}_{k}=(-1)^k\left(\frac{2\pi}{N}\right)^k (\varphi^{K})^{(k)}(\tau),
	\end{align} 
	where $(\varphi^{K})^{(k)}(\tau)$ is the $k$-th order derivative.
\end{lemma}

\begin{proof}
	
	We have
	
		\begin{align}\label{lemma2}
			\Delta^{k}\varphi^{K}_{k}&=(-1)^k\int_{t-\frac{2\pi}{N}}^{t}\int_{x_{k}-\frac{2\pi}{N}}^{x_{k}}\dots\int_{x_{2}-\frac{2\pi}{N}}^{x_{2}}(\varphi^{K})^{(k)}(x_1)dx_1\dots dx_m\nonumber\\
		&=(-1)^k\underbrace{\int_{-\frac{2\pi}{N}}^{0}\dots\int_{-\frac{2\pi}{N}}^{0}}_{k-1}\int_{t-\frac{2\pi}{N}}^{t}(\varphi^{K})^{(k)}(\tilde x_k)d\tilde x_{k}\dots d\tilde x_1\nonumber\\
		&=(-1)^k \left(\frac{2\pi}{N}\right)^{k-1}\int_{t-\frac{2\pi}{N}}^{t}(\varphi^{K})^{(k)}(\tilde x_k)d\tilde x_{k}.
		\end{align} 
	
	By the mean value theorem, there exists $\tau\in\left(t-\frac{2\pi}{N},t\right)$ such that $\int_{t-\frac{2\pi}{N}}^{t}(\varphi^{K})^{(k)}(x)dx=\frac{2\pi}{N}(\varphi^{K})^{(k)}(\tau)$ which completes the proof.

\end{proof}

Now we can proceed with the error estimate.
\begin{proposition} [Error analysis for $m$-th order $\Sigma\Delta$]
	Suppose $f\in L^2(\mathbb T)$ and $\varphi^K\in L^2(\mathbb T)$ satisfies \eqref{kernel:cond}. Suppose $y_n:=f\left(2\pi n/ N\right)$ serves as an input to the $m$-th order $\varSigma\varDelta$ quantizer given by \eqref{sigmadelta} and \eqref{sigmadelta:m} with the stability criterion \eqref{eq:stability} satisfied. Assume that the state variables are initialized as follows: ${u}_{-1}={u}_{-2}=...={u}_{-m}:=0$.
	
	Then, for all $t\in\mathbb T$ and for $m=1$
	\begin{align}\label{Thm1:result0}
	|f(t)-f_r(t)|\leq \frac{\|u\|_{\ell^{\infty}}}{N}\big\|(\varphi^{K})'\big\|_{L^1}+\frac{|u_{N-1}|}{N}\big\|\varphi^{K}\big\|_{L^{\infty}}
	\end{align}
	and for $m\geq 2$
	\begin{align}\label{Thm1:result}
		|f(t)- f_r(t)|&\leq\frac{(2\pi)^{m-1}\|u\|_{\ell^{\infty}}}{N^m}\left(\big\|(\varphi^{K})^{(m)}\big\|_{L^1}+\big\|(\varphi^{K})^{(m-1)}\big\|_{L^{\infty}}\right)\nonumber\\&
		+\frac{1}{N}\sum_{k=1}^{m-1}\left(\frac{2\pi}{N}\right)^{k-1}\big|\Delta^{m-k} u_{N-1}\big|\big\|(\varphi^{K})^{(k-1)}\big\|_{L^{\infty}}.
	\end{align}
\end{proposition}

\begin{proof}
	For the first order quantizer, inequality \eqref{Thm1:result0} is obtained in \eqref{Eq:Rec:Err:1st}. Now we focus on higher order quantizers with $m\geq 2$. Due to the initialization of $u_n$ and the fact that $\Delta^{k-1}\varphi_{N+k-1}=\Delta^{k-1}\varphi_{k-1}$ we can simplify the result in Lemma 1 and get
	\begin{align}\label{Thm1:1}
		&\sum_{n=0}^{N-1}\Delta^m u_n \varphi_n^{K}= (-1)^m\sum_{n=0}^{N-1} u_n\Delta^m \varphi^{K}_{n+m}+ \sum_{k=1}^m (-1)^{k+1} \Delta^{m-k}u_{N-1}\Delta^{k-1}\varphi^{K}_{k-1}.
	\end{align} 
	This yields
	\begin{align}\label{Thm1:2}
		|f(t)- f_r(t)|
		&=\frac{1}{N}\left|\sum_{n=0}^{N-1}\left(y_n-q_n\right)\varphi^K_{n}\right|=\frac{1}{N}\left|\sum_{n=0}^{N-1}\Delta^m u_n\varphi^K_{n}\right|\nonumber\\
		&\leq\frac{1}{N} \left(\sum_{n=0}^{N-1}\left|(-1)^m u_n\Delta^m \varphi^{K}_{n+m}\right|+ \sum_{k=1}^m \big|(-1)^{k+1} \Delta^{m-k}u_{N-1}\Delta^{k-1}\varphi^{K}_{k-1}\big|\right).
	\end{align}
	Now we rewrite the second sum in \eqref{Thm1:2} as two summands and apply Lemma 2. We get
	\begin{align}\label{Thm1:4}
		&\sum_{k=1}^m \big|(-1)^{k+1} \Delta^{m-k}u_{N-1}\Delta^{k-1}\varphi^{K}_{k-1}\big|\nonumber\\
		&=\left|(-1)^{m+1} u_{N-1}\Delta^{m-1}\varphi^{K}_{m-1}\right|+\sum_{k=1}^{m-1} \big|(-1)^{k+1} \Delta^{m-k}u_{N-1}\Delta^{k-1}\varphi^{K}_{k-1}\big|\nonumber\\
		&\leq \|u\|_{\ell^{\infty}}\left(\frac{2\pi}{N}\right)^{m-1}\big| (\varphi^{K})^{(m-1)}(\tau)\big|+\sum_{k=1}^{m-1}\left(\frac{2\pi}{N}\right)^{k-1}\big|\Delta^{m-k} u_{N-1}\big|\big| (\varphi^{K})^{(k-1)}(\tau)\big|\nonumber\\
		&\leq \|u\|_{\ell^{\infty}}\left(\frac{2\pi}{N}\right)^{m-1}\big\|(\varphi^{K})^{(m-1)}\big\|_{L^{\infty}}+\sum_{k=1}^{m-1}\left(\frac{2\pi}{N}\right)^{k-1}\big|\Delta^{m-k} u_{N-1}\big|\big\|(\varphi^{K})^{(k-1)}\big\|_{L^{\infty}}.\nonumber\\
	\end{align}	
	Let us now evaluate the first sum in \eqref{Thm1:2}.
		
	\begin{align}\label{Thm1:3}
		\sum_{n=0}^{N-1}\left|(-1)^mu_n\Delta^m\varphi^K_{n+m}\right|&\leq \|u\|_{\ell^{\infty}}\sum_{n=0}^{N-1}\int_{t-\frac{2\pi(n+1)}{N}}^{t-\frac{2\pi n}{N}}\int_{x_{m}-\frac{2\pi}{N}}^{x_{m}}\dots\int_{x_{2}-\frac{2\pi}{N}}^{x_{2}}\big|(\varphi^{K})^{(m)}(x_1)\big|dx_1\dots dx_m\nonumber\\
		&=\|u\|_{\ell^{\infty}}\int_{t-2\pi}^{t}\int_{x_{m}-\frac{2\pi}{N}}^{x_{m}}\dots\int_{x_{2}-\frac{2\pi}{N}}^{x_{2}}\big|(\varphi^{K})^{(m)}(x_1)\big|dx_1\dots dx_m\nonumber\\
		&=\|u\|_{\ell^{\infty}}\underbrace{\int_{-\frac{2\pi}{N}}^{0}\dots\int_{-\frac{2\pi}{N}}^{0}}_{m-1}\int_{t-2\pi}^{t}\big|(\varphi^{K})^{(m)}(\tilde x_m)\big|d\tilde x_{m}\dots d\tilde x_1\nonumber\\
		&=\|u\|_{\ell^{\infty}}\left(\frac{2\pi}{N}\right)^{m-1}\big\|(\varphi^{K})^{(m)}\big\|_{L^1}.
	\end{align}
		
	Inserting \eqref{Thm1:3} and \eqref{Thm1:4} into \eqref{Thm1:2} completes the proof.	
\end{proof}

Essentially, Proposition 2 tells us that whilst the first term of the error estimate \eqref{Thm1:result} gives the error of $O(N^{-m})$ for the $m$-th order $\Sigma\Delta$ quantizer, the second term always gives a sum of larger errors of order up to $O(N^{-1})$, hence increasing the order of the  $\Sigma\Delta$ scheme without making any changes to the existing scheme becomes meaningless.

\section{Modification of the $\Sigma\Delta$ schemes}

The question to address in this section is how to improve the reconstruction error \eqref{Thm1:result0} and \eqref{Thm1:result} by modifying the $m$-th order $\Sigma\Delta$ modulation scheme described in Section II.A.

We propose to use the recurrence relation \eqref{sigmadelta} once again on updated samples
\begin{equation}\label{update}
	\tilde y_n:=y_n+\delta,
\end{equation}
where some small constant $\delta$ independent of $n$ is added at each iteration. We denote the resulting updated variables $\tilde u_n$ and $\tilde q_n$. We require that $\tilde y_n$ are the samples of a $K$-bandlimited function $\tilde f(t)$ and denote the error between this function and its reconstruction $\tilde e(t):=|\tilde f(t)-\tilde f_r(t)|$. 
In what follows, we will show that using an appropriate constant $\delta=\mathcal{O}(\tfrac{1}{N})$ improves  $\tilde e(t)$ for the $m$-th order $\Sigma\Delta$ quantizer with $m=1,2$. In particular, in case $m=2$ this update leads to the error $\tilde e(t)$ of $O(N^{-2})$. 

Clearly, after the update $\delta$, one actually considers samples of the function $\tilde f=f+\delta$, so one can at best hope to recover that function. For our choice $\delta=\mathcal{O}(\tfrac{1}{N})$, we hence recover $f$ up to an error of order $\mathcal{O}(\tfrac{1}{N})$, just as the first order scheme. We argue, however, that a constant error is much preferred to an arbitrary error of the same maximal amplitude. The latter can be perceived as an artefact, whereas the constant update will be hardly noticeable. In fact, as the following lemma shows, an error of order $\mathcal{O}(\tfrac{1}{N})$ can, in general not be avoided even when restricting to very small amplitudes.
\begin{lemma}
Fix $\alpha>0$. Then  for all $\tfrac{1}{\alpha}\leq N\in\mathbb N$ there exists $f$ $K$-bandlimited with $\|f\|_\infty \leq \alpha$ 
such that for any choice of $q\in \{-1,1\}^N$ of bit sequences, the associated reconstruction $f_q(t)= \frac{1}{N}\sum_{n=0}^{ N-1}q_{n}\varphi^K\left(t-\frac{2\pi n}{ N}\right)$ satisfies $\big\|f-f_q\big\|_{L^{\infty}} \geq \tfrac{1}{2N}$.
\end{lemma}
\begin{proof} Averaging over $N$ samples of $f-f_q$ yields that for arbitrary $f$ and $q$
    \begin{align}\label{rvsz}
	\big\|f-f_q\big\|_{L^{\infty}}
	&\geq \frac{1}{N^2} \left|\sum_{k=0}^{ N-1}\sum_{n=0}^{ N-1}\left(y_{n}-q_{n}\right) \varphi^K\left(\frac{2\pi k}{ N}-\frac{2\pi n}{ N}\right)\right|\nonumber\\
	&= \frac{1}{N} \left|\sum_{n=0}^{ N-1}\left(y_{n}-q_{n}\right)\frac{1}{N}\sum\nolimits_{k=0}^{ N-1} \varphi^K\left(\frac{2\pi k}{ N}-\frac{2\pi n}{ N}\right)\right|\nonumber\\
	&= \frac{1}{N} \left|\sum_{n=0}^{ N-1}y_{n}-\sum_{n=0}^{ N-1}q_{n}\right|.
\end{align}
The expression $\sum\nolimits_{n=0}^{ N-1}y_{n}-\sum\nolimits_{n=0}^{ N-1}q_{n}$ is generically different from zero; in particular, choosing $f$ to be the constant function with value $\tfrac{1}{2N}<\alpha$ yields that $\sum\nolimits_{n=0}^{ N-1}y_{n}=\tfrac{1}{2}$, while $\sum\nolimits_{n=0}^{ N-1}q_{n}\in \mathbb{Z}$ and hence $ \big\|f-f_q\big\|_{L^{\infty}}\geq \tfrac{1}{2}$, as desired.
\end{proof}

We now seek to design a constant update such that the right hand side of \eqref{rvsz} vanishes.
Using the initialization ${u}_{-1}={u}_{-2}=...={u}_{-m}:=0$ and computing the telescoping sum, we have 
\begin{equation}
\sum_{n=0}^{N-1}\Delta^m u_n=\Delta^{m-1}u_{N-1}
\end{equation}
which together with \eqref{sigmadelta:m} yields
\begin{equation}\label{remainder}
	\sum_{n=0}^{N-1}y_n-\sum_{n=0}^{N-1}q_n=\Delta^{m-1}u_{N-1}.
\end{equation}
Thus one can achieve a lower bound equal to zero in \eqref{rvsz} by adding a constant update $\delta$ to each $y_n$ with $N\delta=-\Delta^{m-1}u_{N-1}$. 


Although the error cannot be avoided in general, we can argue that this is a reasonable trade-off. 
Namely an error caused by a constant update can be less audible/visible than highly oscillating initial error. Numerical simulations in Section V show that the distribution of the error $f(t)-\tilde f_r(t)$ around the circle is more uniform than the distribution of the error $f(t)-f_r(t)$.

\subsection{First order $\varSigma\varDelta$ scheme}

In this subsection, we will show that for the first order $\Sigma\Delta$ scheme it is sufficient to choose a constant update
\begin{equation}\label{update1}
	\delta=-N^{-1}u_{N-1}
\end{equation}
in order to eliminate the boundary term of summation by parts in the error estimate \eqref{Thm1:result0}. Namely, we will prove that this update causes $\tilde u_{N-1}=0$. We will start with the following lemma.

\begin{lemma}\label{Lem:1} Let $a\in [0,1]$ and $y_n\in[-1-a,1+a]$ for all $n$. Then for recurrence relation \eqref{sigmadelta:m} with $m=1$ and ${u}_{-1}:=0$ it holds that $u_n\in[-1-na,1+na]$.
\end{lemma}
\begin{proof}
	We prove lemma by induction. Note that for $\Sigma\Delta$ quantizer with $m=1$ we have $v_n=u_n$. Base case: For $y_0>0$ it follows from \eqref{sigmadelta:m} with $m=1$ that $u_0=y_0-1$ and thus $u_0\in(-1,a]$. Similarly, for $y_0\le0$ we have $u_0\in [-a,1]$. As $a\in [0,1]$, it follows that $u_0\in[-1,1]$.\\
	Induction step: Assume $u_n\in[-1-na,1+na]$. Then, $|u_n+y_{n+1}|\le 2+(n+1)a$. We distinguish again two cases: if $u_n+y_{n+1}> 0$, one has $u_{n+1}=u_n+y_{n+1}-1\in (-1,1+(n+1)a]$ and analogously  $u_n+y_{n+1}\le0$ implies that $u_{n+1}\in [-1-(n+1)a,1]$.
	As $a$ is nonnegative, we have $|u_{n+1}|\le1+(n+1)a$.
\end{proof}

\begin{theorem}[Uniform update for 1-st order $\Sigma\Delta$]\label{Thm:v:1st}
		Let the assumptions in Proposition 2 hold and assume the order of $\varSigma\varDelta$ quantizer $m=1$. Then, using in \eqref{sigmadelta} the updated samples $\tilde y_n:=y_n+\delta$ with $\delta$ given by \eqref{update1} leads to $\tilde u_{N-1}=0$.
\end{theorem}
\begin{proof}
	By assumptions in Proposition 2, we have $y_n\in [-1,1]$ for the first order $\Sigma\Delta$ quantizer. By Lemma \ref{Lem:1}, we have $u_n\in[-1,1]$ for all $n$ and thus $\tilde y_n\in[-1-\frac 1N,1+\frac 1N]$. Using Lemma \ref{Lem:1} again, we obtain that $\tilde u_n\in [-1-\frac nN,1+\frac nN]$ and, in particular, $\tilde u_{N-1}\in [-1-\frac {N-1}N,1+\frac{N-1}N]\subset (-2,2)$. 
	
	On the other hand, 
	\begin{align}\label{Eq:v}
			\tilde u_{N-1}&=\sum_{n=0}^{N-1}(\tilde y_n-\tilde q_n)=\sum_{n=0}^{N-1}(y_n -\frac{u_{N-1}}N-\tilde q_n)\nonumber\\
			&=\sum_{n=0}^{N-1}(y_n -q_n)-\frac{u_{N-1}}{N}N+ \sum_{n=0}^{N-1}(q_n-\tilde q_n)\nonumber\\
			&=\sum_{n=0}^{N-1}q_n-\sum_{n=0}^{N-1}\tilde q_n.
	\end{align}
	Here, we used equation \eqref{remainder} for the last equality.
	Denote the number of $+1$ in the vector $(q_n)_n$ by $L$ and the corresponding number for $(\tilde q_n)_n$ by $\tilde L$. Then, we conclude from equation \eqref{Eq:v} that
	\begin{equation}\label{el}
		\tilde u_{N-1}=2L-N -(2\tilde L -N)= 2(L-\tilde L).
	\end{equation}
	In particular, $\tilde u_{N-1}$ has to be an even integer in the interval $(-2,2)$, i.e., $\tilde u_{N-1}=0$.
\end{proof}

\subsection{Second order $\varSigma\varDelta$ scheme}


From \eqref{Thm1:result} we see that in the second order case there is only one summand larger than of $O(N^{-2})$ and it involves the remainder $|\Delta u_{N-1}|$. Therefore, we choose a constant update similar to that in Section IV.A, namely,
\begin{equation}\label{update2}
	\delta=-N^{-1}\Delta u_{N-1}.
\end{equation}
The following theorem ensures that this approach is valid.

\begin{theorem}[Uniform update for $2$-nd order $\Sigma\Delta$]
	Let the assumptions in Proposition 2 hold and assume the order of $\varSigma\varDelta$ quantizer $m=2$. Let the updated samples $\tilde y_n:=y_n+\delta$ with $\delta$ given by \eqref{update2} satisfy 
	the stability criterion \eqref{eq:stability}. 
	Then, using $\tilde y_n$ in \eqref{sigmadelta} leads to $\Delta \tilde u_{N-1}=0$ and the error $\tilde e(t):=|\tilde f(t)-\tilde f_r(t)|$ is of $O(N^{-2})$.
\end{theorem}

\begin{proof}
	
	By a simple inductive argument we show that  $\tilde v_{n{\tiny }}\in(-1,1)$. Namely, suppose $\tilde v_{n-1},\dots,\tilde v_{n-k}$ $\in(-1,1)$ for a $k$-tab filter. Then $|(h\ast \tilde v)_{n}+\tilde y_n|< \|h\|_{\ell^1} + \|\tilde y\|_{\ell^{\infty}} \leq 2$. The last inequality follows from the stability criterion \eqref{eq:stability}. It then follows that $|\tilde v_n|=|(h\ast \tilde v)_{n}+\tilde y_n - \mathrm{sign} ((h\ast \tilde v)_{n}+\tilde y_n)|<1$. 
	
	Now we move to the state variables $\tilde u_n$. From the construction \eqref{filter2} of the filter and from \eqref{sigmadelta:m2} it follows that 
	$\Delta^2 g=(1,-k/(k-1),0,\dots,0,-1/(1-k))$ and $\Delta g=(1,-1/k,\dots,-1/k,0)$, therefore $\|\Delta g\|_{\ell^1}=2$. We have $|\Delta\tilde u_n|=|(\Delta g\ast \tilde v)_n|< \|\Delta g\|_{\ell^1}=2$ and, in particular, $\Delta \tilde u_{N-1}\in(-2,2)$.
	
	Similar to \eqref{el} in Theorem \ref{Thm:v:1st}, we have $\Delta \tilde u_{N-1}= 2(L-\tilde L)$ and conclude that updating with $\delta_n=-N^{-1}\Delta u_{N-1}$ leads to $\Delta \tilde u_{N-1}=0$. The error estimate is obtained by rewriting \eqref{Thm1:result} with updated variables,
	\begin{align}
		|\tilde f(t)- \tilde f_r(t)|&\leq\frac{2\pi\|u\|_{\ell^{\infty}}}{N^2}\left(\big\|(\varphi^{K})''\big\|_{L^1}+\big\|(\varphi^{K})'\big\|_{L^{\infty}}\right).
	\end{align}
\end{proof}

The updated second order scheme has a clear advantage over the updated first order scheme. While the updated first order scheme yields at best an error of $O(N^{-1})$ which is already achieved by the classical scheme, the updated second order scheme yields $O(N^{-2})$ which is not possible to obtain without an update as long as $\sum\nolimits_{n=0}^{ N-1}y_{n}\neq\sum\nolimits_{n=0}^{ N-1}q_{n}$.

\subsection{Outlook: higher order $\varSigma\varDelta$ schemes}
Constant updates similar to \eqref{update1} and \eqref{update2} can be constructed for a $\Sigma\Delta$ scheme of arbitrary order. For the $m$-th order scheme we can have
\begin{equation}\label{updatem}
\delta=-N^{-1}\Delta^{m-1} u_{N-1}
\end{equation}
leading to $\Delta^{m-1} \tilde u_{N-1}=0$.
Apart from setting the lower bound to zero in \eqref{rvsz}, this update does not provide much improvement for higher order schemes. For any scheme of order $m>2$, the error $\tilde e(t)$ will be still of $O(N^{-2})$.

For instance, in the third order case before the update we have two unwanted large terms of $O(N^{-1})$ and $O(N^{-2})$ in \eqref{Thm1:result} involving the remainders $|\Delta^2 u_{N-1}|$ and $|\Delta u_{N-1}|$, respectively. Here, we cannot achieve both  $\Delta^2 \tilde u_{N-1}=0$ and $\Delta \tilde u_{N-1}=0$ 
by a constant update and get $O(N^{-2})$ at best. In \cite{GKKS2019}, we proposed a non-constant, slowly varying sequence $\delta_n$ of updates sampled from an appropriately shifted and scaled sinusoid. Despite allowing for $\tilde e(t)$ of $O(N^{-3})$, this approach has several drawbacks. Firstly, it requires additional assumptions on the differences $q_n-\tilde q_n$ which obstruct the development of a solid theory. Secondly, the scaling factor of the sinusoid is of $O(\tan(\frac{\pi}{N}))$ which is close to $O(N^{-1})$. In applications such as digital halftoning, these introduced oscillations may be more visible to the human visual system than a constant update.


\section{Numerical experiments}

For our numerical experiments we consider the signal 
$f(t)=0.1\sin(5t)\cos(10t)+0.2$ with bandwidth $K=15$ and the reconstruction formulas \eqref{shannonquantized_circle} and \eqref{repkernel}. 

Fig. 1.(a) shows the reconstruction error $\tilde e(t):=|\tilde f(t)-\tilde f_r(t)|$ for the updated $\Sigma\Delta$ schemes of order $m=1,2$. We can observe the desired error decay analogous to the error decay of classical schemes on the real line. However, the error $\tilde e(t)$ only takes the shifted version of the original function into account. In Fig. 1.(b) and Fig. 1.(c) we show the difference between $f(t)$, the original function, and $f_r(t)$, its reconstructed version. We also show the difference between $f(t)$ and $\tilde f_r(t)$, the function reconstructed from updated samples. Due to the boundary term elimination caused by the update, the large error at the initialization point $0$ (or $2\pi$) is reduced. Note that while for the classical scheme the error keeps oscillating around zero, for the updated scheme it keeps oscillating around $\delta$, as constant shifts of $O(N^{-1})$ cannot be avoided. Nevertheless, the updated scheme optimizes oscillation amplitude around a constant shift, thus leading to less visible artifacts.


\begin{figure}[ht]
		\begin{minipage}[b]{0.8\columnwidth}
			\centering
			\includegraphics[width=1\columnwidth]{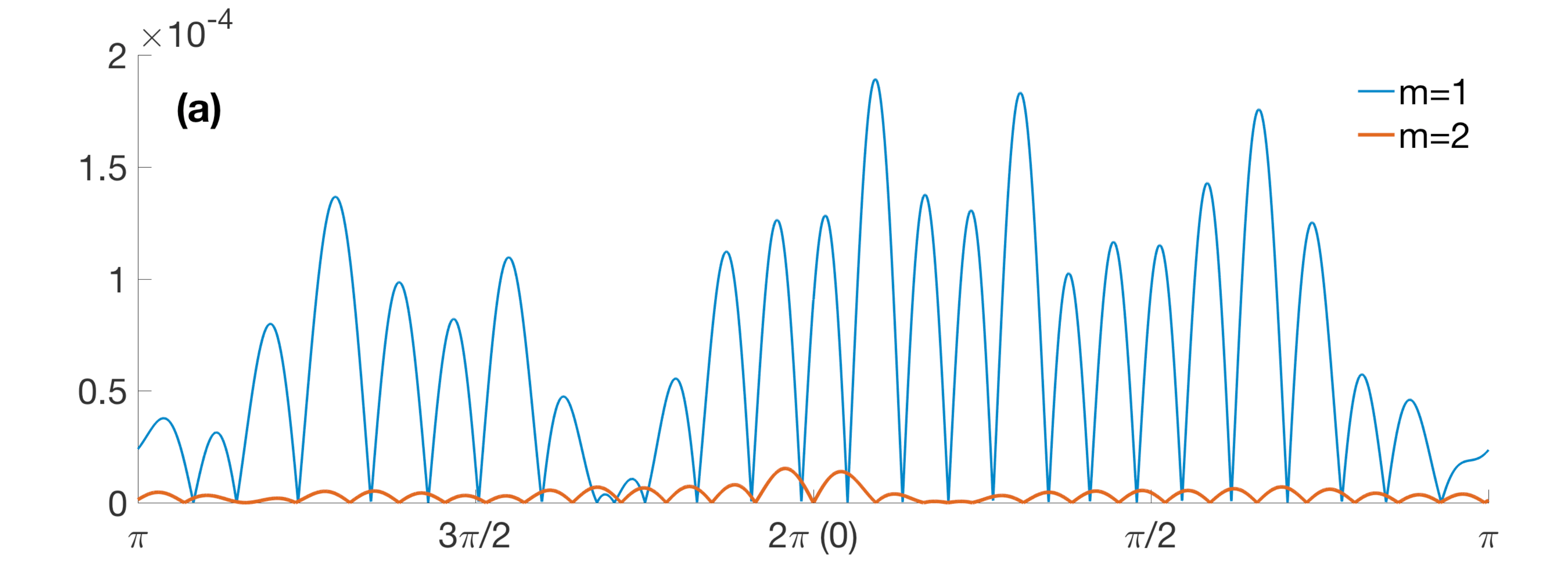}
		\end{minipage}
	\begin{minipage}[b]{0.8\columnwidth}
		\centering
		\includegraphics[width=1\columnwidth]{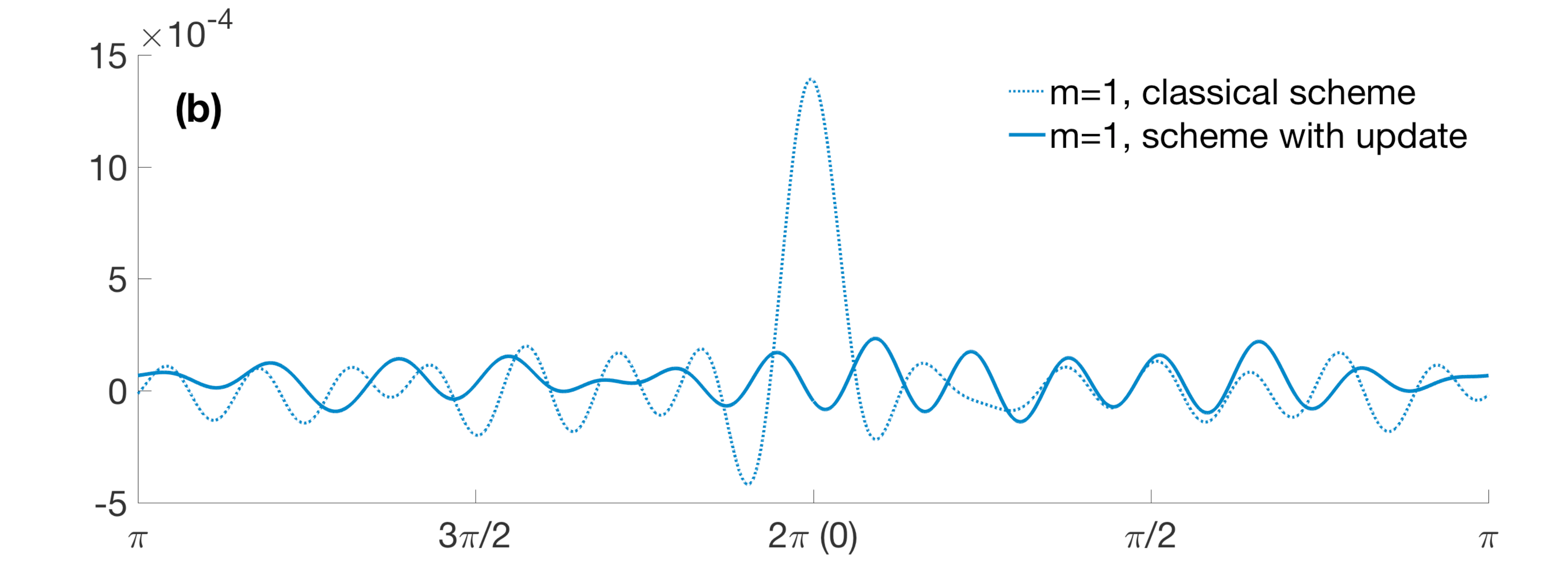}
	\end{minipage}
	\begin{minipage}[b]{0.8\columnwidth}
		\centering
		\includegraphics[width=1\columnwidth]{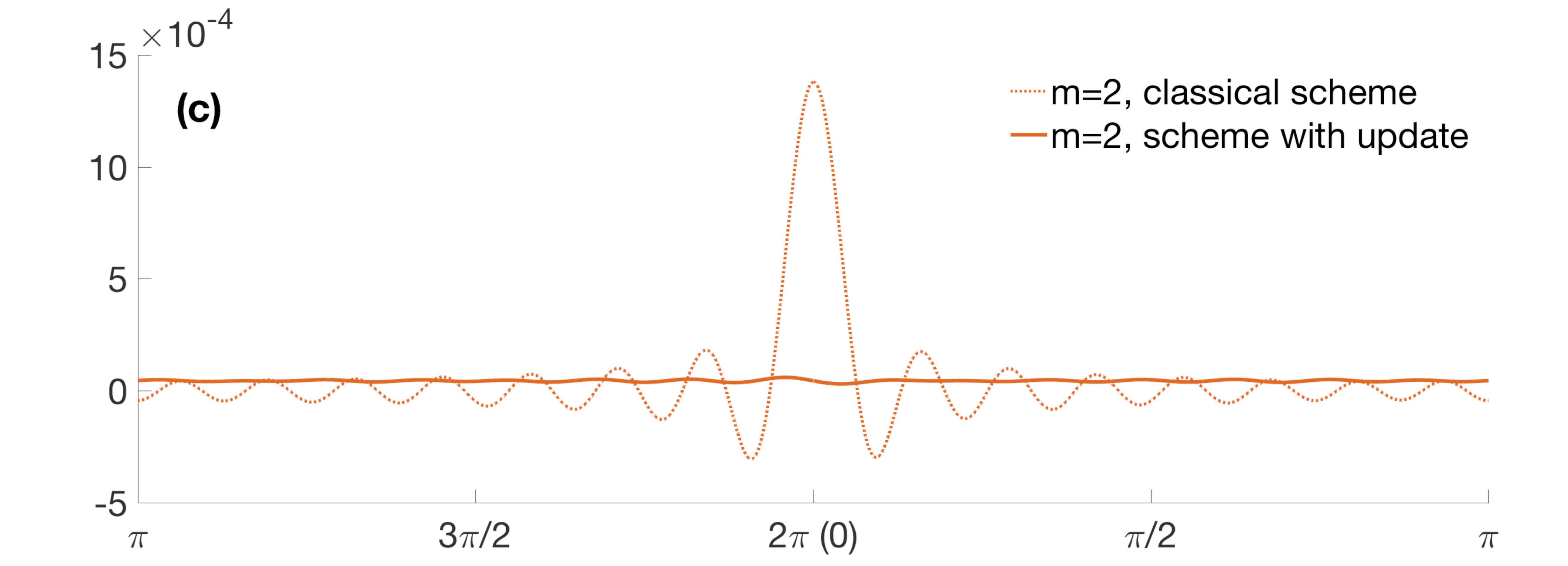}
	\end{minipage}
	\centering
	
	\caption{Reconstruction error for $\Sigma\Delta$ schemes on the circle. In all experiments we use $f(t)=0.1\sin(5t)\cos(10t)+0.2$ with bandwidth $K=15$, $N=9002$ and feedback filters $(0,1)$,  $(0,4/3,0,0,-1/3)$ for the 1-st and the 2-nd order schemes, respectively. \textbf{\textsf{(a)}} Comparison of errors $\tilde e(t):=|\tilde f(t)-\tilde f_r(t)|$ for updated $\Sigma\Delta$ schemes of order $m=1,2$. \textbf{\textsf{(b)}} Errors $f(t)-f_r(t)$ and $f(t)-\tilde f_r(t)$ for the 1-st order $\Sigma\Delta$ scheme (cf. Theorem 1). \textbf{\textsf{(c)}} Errors $f(t)-f_r(t)$ and $f(t)-\tilde f_r(t)$ for the 2-nd order $\Sigma\Delta$ scheme (cf. Theorem 2).}
\end{figure}

\section{Conclusions and future work}
In this paper, we justified the necessity of updating the classical $\Sigma\Delta$ modulation scheme if the function to be quantized is defined on the circle. We proposed the updates for the first and the second order $\Sigma\Delta$ schemes and complemented our results with the reconstruction error analysis.
For the schemes of order $m\geq3$, designing the update becomes less trivial. The important open question is how to find an optimal update for a $\Sigma\Delta$ scheme of arbitrary order. Furthermore, it is of paramount interest to extend the results to more complicated manifold models.

\pagebreak

\bibliographystyle{plain}
\bibliography{refs}
\end{document}